\documentclass[11pt]{article}
\usepackage{style}

\usepackage{epsf}
\usepackage{epsfig}
\usepackage{graphicx}
\usepackage{color}
\usepackage{mathrsfs}

\begin{document}
	
\title{\textbf{Shor's Algorithm Does Not Factor Large Integers in the Presence of Noise}}
\author{
Jin-Yi Cai \\ {University of Wisconsin-Madison} \\ {\tt jyc@cs.wisc.edu}
\\
\\
}
\date{}
\maketitle

\pagenumbering{gobble}

\begin{abstract}
We consider Shor's quantum factoring algorithm in the setting of
noisy quantum gates. Under a 
generic model of random noise
for (controlled) rotation gates, we prove that the algorithm does not
factor integers of the form $pq$ when the noise exceeds a vanishingly small
level in terms of $n$ --- the number of bits of the integer to be factored,
 where $p$ and $q$
are from a well-defined  set of primes of positive density.
We further prove that with probability $1 - o(1)$
over random prime pairs $(p,q)$,
Shor's factoring algorithm does not factor numbers of the form $pq$,
with the same level of random noise present.
%
\end{abstract}

\pagenumbering{arabic}
\setcounter{page}{1}
\section{Introduction}\label{sec:introduction}

One of the most stunning achievements of computer science in the last several decades
is Shor's quantum algorithm to factor large integers~\cite{shor-conf,shor}.
The algorithm provably can factor an $n$-bit integer in polynomial time
with high probability, assuming certain quantum operations can be performed.
These are called quantum logic gates. In particular, they include
the familiar Hadamard gate $H =\frac{1}{\sqrt{2}} \left[\begin{smallmatrix} 1 & 1 \\
1&-1 \end{smallmatrix}\right]$, the rotation gates (Phase) 
$S =  \left[\begin{smallmatrix} 1 & 0 \\ 0 & i \end{smallmatrix}\right]$,
($\pi/8$ gate) $T =  \left[\begin{smallmatrix} 1 & 0 \\ 0 & e^{2\pi i/8} \end{smallmatrix}\right]$,
and more generally $R_k =  \left[\begin{smallmatrix} 1 & 0 \\ 0 & e^{{2 \pi i}/{2^k}} \end{smallmatrix}\right]$,
and their controlled versions. Note that $S =  R_2$ and $T = R_3$.

It has often been pointed out that the availability of these quantum gates at high
precision (with arbitrarily small angles in $R_k$ with $k \rightarrow \infty$)
is a challenge, both intellectually and practically on engineering grounds~\cite{gil,levin,Landauer,aaronson}. To a large extent, such concerns motivated
another great  intellectual achievement that is the development of
quantum error correcting codes~\cite{shor2,steane,Calderbank-4,Gottesman-intro-code-arxiv,book}. 
There is a substantial body of work on fault tolerant quantum computing, 
starting with Shor's work~\cite{shor-fault-tolerant}. Strong threshold theorems
are proved which show that in certain error models, 
if the error rate is below a certain threshold, quantum computation
can achieve arbitrarily high accuracy~\cite{Aharonov-Ben-Or,Kitaev,Aliferis,Gottesman-intro-code-arxiv,Steane-263,Knill-178,Aliferis-Gottesman-Preskill-19}.
These are beautiful mathematical theorems. But they fundamentally
assume that the group ${\rm SU}(2)$ \emph{exactly} corresponds to operations on 
a qubit in reality,
especially in its composition---that group composition (in its
infinite precision defined over $\mathbb{C}$) exactly corresponds
to sequential application of realizable quantum operations.
Opinions differ, as to whether such arbitrary precision is ever achievable.
It is certainly a possibility.
However, this author is skeptical about this, based on
the belief that quantum mechanics itself (just as any other physical theory) is 
not, and is not meant to be, 
infinitely accurate when comparing reality with what the mathematical
statements say in the theory (some speculations are in Section~\ref{Comments.tex}).
%
Meanwhile, enormous efforts have been underway in
the past few decades, and with much renewed momentum and enthusiasm more recently,
  to achieve ever increasingly accurate hardware implementations of quantum circuitry.

In this paper, we consider Shor's quantum factoring algorithm in
the setting where each quantum controlled rotation gate is
subject to a small random noise in the angle. We assume
each application of the controlled-$R_k$ gate is given an independent random error
of angle $e^{{2 \pi i} \epsilon r/{2^k}}$.
Thus, when the control bit is 1, the operator $R_k$ is substituted by 
$\widetilde{R_k} = 
 \left[\begin{smallmatrix} 1 & 0 \\
0 & e^{{2 \pi i( 1 + \epsilon r)}/{2^k}} \end{smallmatrix}\right]$,
 where $r$ is an \emph{independent} noise random variable
 distributed $r \sim N(0, 1)$, and $ \epsilon$ is a global magnitude parameter.
So, the controlled-$\widetilde{R_k}$ gate is
$\left[\begin{smallmatrix} 1 & 0 & 0 & 0  \\ 
0 & 1 & 0 & 0 \\
0 & 0 & 1 & 0\\
0 & 0 & 0 & \rho_k \xi_k \end{smallmatrix}\right]$, where
$\rho_k = \rho_{k,  \epsilon} =  e^{{2 \pi i \epsilon r}/{2^k}}$ and 
$\xi_k = e^{{2 \pi i}/{2^k}}$.
We show that there exist positive constants $c, c' >0$ such that
if $\epsilon > c n^{-1/3}$, 
then  Shor's  algorithm  does not factor
$n$-bit integers of the form $pq$, where $p$ and $q$
are from a well-defined  set of primes of density $> c'$.
To the best knowledge of this author, this is the first provable statement
of such failure of Shor's algorithm under any error model.

The noise model is similar to that of~\cite{Nam-B-2014}
(see also~\cite{Fowler,Nam-B-2013,Nam-B-2015}). The specific random noise model
including
the independent normal distribution picked in this paper
is not essential, as the proof will clearly show, but it is chosen
to present the essential idea of the proof most transparently. 
For example, the noise r.v.~$r$ being distributed $\sim N(0, 1)$
can be replaced by any reasonable alternative distribution such as
uniform  $U[-1, 1]$ or uniform bits from $\{-1, 0, 1\}$. 
While each individual controlled-$R_k$ gate is assumed to be
accompanied by an independent
r.v. $r$ for noise, when an individual controlled-$R_k$ gate is applied,
the same randomly perturbed  controlled-$R_k$ gate is 
applied to each term in a sum of superpositions of quantum states.
Regarding the random noise model,
we do not make any claim that this model accurately reflects 
``reality''; our purpose is  only to show that \emph{some} vanishing amount of
noise can already provably destroy the algorithm.

An important modification of Shor's algorithm by Coppersmith~\cite{Coppersmith}
shows that if we just ignore (not to perform) all (controlled-) $R_k$-gates
for sufficiently large $k \ge b$, where $b$ is some global parameter,
then  Shor's algorithm still retains its
effectiveness (and uses a reduced number of quantum gates).
The specific suggested change~\cite{Coppersmith} for $500$-qubits, which would require 
rotations of magnitude $2 \pi/2^{500}$ in Shor's original algorithm,
is to ignore all rotations of angle smaller than $2 \pi/2^{20}$.
It is estimated that this would incur an error on the order of $1\%$
in the probability of each desirable final state. Asymptotically,
Coppersmith improves the precision requirement of exponentially small angles
to just slightly less than $\pi/n$. This is of enormous practical implications.
This version of Shor's algorithm is called the ``banded'' version with parameter $b$,
which is set to be slightly greater than $\log n$, rather than $n$ in the
original version.
Nonetheless, rotation gates (as primitive steps
of the algorithm) of asymptotically infinitely small angles
would still be required as $n$, the number of bits to be factored, tends to
infinity. 

Our result is consistent with Coppersmith's improvement.
Indeed we will present our proof in the ``banded'' version,
with perfect controlled-$R_k$-gates for all
 $k < b$, but  every controlled-$R_k$-gate is replaced by
 a controlled-$\widetilde{R_k}$-gate for all $k \ge b$, i.e., it
is  independently  perturbed by
 a  random noise. 
Our negative result will be stated in terms of $b+ \log_2 ( 1/\epsilon)$.
When  $b+ \log_2 \left(1/\epsilon\right) < \frac{1}{3} \log_2 n - c$
for some constant $c>0$, the noise takes hold
so as to destroy the desired peak in the probability of observing
a useful state that leads to factorization. This condition is essentially equivalent
to having both $b$ being less than a small constant multipple of
 $\log n$ 
and $\epsilon$ greater than the reciprocal of a  small positive power  of $n$.
We prove that, under this condition in this noise model,
 Shor's  algorithm  does not factor
$n$-bit integers of the form $pq$, where $p$ and $q$
are from a well-defined set of primes of positive density $c' >0$.~\footnote{We note that
the results from~\cite{Fowler,Nam-B-2013,Nam-B-2014,Nam-B-2015} are generally
stated in the opposite direction. Under plausible, but ultimately heuristic,
 assumptions for the bahavior
of various sums, augmented by numerical simulations, they suggest that
if $b$ is not too large compared to $n$, Shor's  algorithm
\emph{can} tolerate imprecisions of rotation angles.
Some small concrete values of $n$ are on the order of 10-qubits ($n= 10, 14$).
These values are quite outside the range where our proof applies.
Their numerical simulation does seem to suggest a logarithmic threshold of $b$.
Thus, these positive
results are not logically inconsistent with, and in fact, complement
 our proof. N.B. the notation $b$ in~\cite{Nam-B-2013} is our $b-2$.}
The proof will in fact show that, the same result holds under the same condition 
 $b+ \log_2 \left(1/\epsilon\right) < \frac{1}{3} \log_2 n - c$,
 even if
the noise gates are applied only at the single level $R_b$, with all
other controlled-$R_k$-gates applied perfectly for $k \ne b$ (\emph{or alternatively},
no controlled-$R_k$-gates are applied at all for $k>b$ as in the
banded version by Coppersmith).

\begin{theorem}\label{main1}
There exist constants $c, c' >0$, such that
if each controlled-$R_k$-gate in the quantum Fourier transform circuit is replaced by
controlled-$\widetilde{R_k}$-gate for all $k \ge b$, where
$b+ \log_2 \left(1/\epsilon\right) < \frac{1}{3} \log_2 n - c$,
then with exponentially small exceptional probability,
Shor's algorithm does not factor
$n$-bit integers of the form $pq$, where $p$ and $q$
are from a well-defined set of primes of density $> c'$.
\end{theorem}

Here ``exceptional probability''
is over
the random choices of Shor’s algorithm as well as probabilistic
outcomes of quantum measurements. 
More precisely, the expectation over random noise $r$'s,
of the success probability (over the random choices of the algorithm
and quantum measurements) of the algorithm is exponentially small in $n$.
This will be the meaning of
 ``does not factor'' below. 


\begin{theorem}\label{main2}
If $b+ \log_2 \left(1/\epsilon\right) < \frac{1}{3} \log_2 n - c$,
then the statement in Theorem~\ref{main1} still holds,
if only each  controlled-$R_b$-gate
is  replaced by a 
controlled-$\widetilde{R_b}$-gate while all other
 controlled-$R_k$-gates remain unchanged.
Alternatively, the same statement holds
if each  controlled-$R_k$-gate
is: (1) applied perfectly for $k<b$, (2) replaced by a
controlled-$\widetilde{R_b}$-gate for $k=b$, and (3) deleted for $k>b$.
\end{theorem}

Our proof focuses on the essential ``period-finding'' 
part using quantum Fourier transform (QFT) in
Shor's  algorithm.
In our proof, we use a theorem of Fouvry~\cite{Fouvry}.
%
%
%
This theorem states that
the set of all primes $p$ such that the
largest prime factor in $p-1$ is greater than $p^{2/3}$
has positive density among all primes.
We use this theorem to produce candidate inputs of the form $N= pq$
 to Shor's  algorithm  where $p$ and $q$ are of this type, 
and argue that a random element $x \in \mathbb{Z}_{N}^*$ has  (exponentially)
large order  $\omega  = \omega_N(x)$ as an element of the multiplicative group $\mathbb{Z}_{N}^*$.
This large  order    $\omega$ allows us to give a lower bound
for a lattice counting argument, which leads to a sufficiently large
number of independent perturbations in the complex arguments (in the exponent) 
in a crucial sum of exponentials,
(which would have been a perfect geometric sum without noise) in the analysis of
Shor's algorithm. This perturbation, at the appropriate setting of
parameters, destroys this geometric sum, and degrades the probability of
observing any useful quantum state to \emph{negligible}, and thus fails to gain any
useful information on the period  $\omega$.

Our proof is actually more generally applicable.
In an appendix we prove the following theorem:

\begin{theorem}\label{main3}
There exists a constant $c>0$, such that
for
random  primes $p$ and $q$ chosen  uniformly from
all primes of binary length $m$,
if  $b+ \log_2 \left(1/\epsilon\right) < \frac{1}{3} \log_2 m - c$,
as $m \rightarrow \infty$
with probability $1- o(1)$, 
Shor's algorithm with noisy rotation gates does not factor
$N= pq$.
\end{theorem}
A version analogous to Theorem~\ref{main2} also holds for random primes.

We make a few brief remarks. 
Arguably, factoring integers $N=pq$ for random primes 
$p$ and $q$ is more important in cryptography than for primes that
satisfy the property in Fouvry's theorem,
and the statement of failure probability being  $1- o(1)$ is stronger than
that of positive density guaranteed by Fouvry's theorem.
We present the proof in the main text for the latter,
and relegate the proof of Theorem~\ref{main3}
to the appendix, in order to
 concentrate on the main idea of how random noise degrades the performance
of Shor's algorithm.
The additional work needed for Theorem~\ref{main3}
is mainly of a number theoretic nature,
and for the 
purpose of this paper, of secondary importance.
Also, one can prove other versions of Theorem~\ref{main3}. E.g.,
we can restrict the random primes $p$ and $q$ to be
of length $m$ \emph{and} both $\equiv 3 \bmod 4$, so that 
the numbers $N=pq$ are the so-called
Blum integers, which are favored in cryptography~\cite{cryp-Handbook}.
Despite the strong failure demonstrated by the proof,
our theorems
do not
rule out  the possibility that at some future time,
quantum algorithm is superior to the best ``classical'' factoring algorithms
for factoring integers of a certain
size, in practice. 
But our proof indicates that there is a limit to this possible superiority
when $n$ is large, if arbitrarily small random noise  cannot be eliminated.

Many people have made strong arguments~\cite{book} supporting the viewpoint that
Shor's algorithm presents a convincing evidence
that the so-called
\emph{Strong} Church-Turing thesis needs a necessary modification. This
\emph{Strong} thesis identifies 
efficient computation with P or BPP. The argued-for modification
states that this should be replaced by BQP.
 This author is personally not convinced of this. I will make some comments at
the end of this paper.~\footnote{These  comments are speculative,
and should not be conflated with the theorems proved in the paper.}

\section{Preliminaries}\label{sec:pre}
\paragraph{Fouvry's theorem}
Let $N = pq$, where $p$ and $q$ are distinct odd primes.  By the Chinese
remainder theorem, the multiplicative group $\mathbb{Z}_N^*
= \{ m \in \mathbb{Z}_N \mid \gcd(m, N) =1\}$ (invertible
elements in $\mathbb{Z}_N$) is isomorphic to the direct product
$\mathbb{Z}_p^{*} \times \mathbb{Z}_q^{*}$. Moreover,  $\mathbb{Z}_p^{*}$
is a cyclic group of order $p-1$, and is isomorphic to a direct product
of factors according to the prime factorization of $p-1$; and similarly
for $\mathbb{Z}_q^{*}$.  If $p-1 = 2^e p_1^{e_1} \cdots p_k^{e_k}$,
where $p_1 < \ldots < p_k$ are distinct odd primes, then  $\mathbb{Z}_p^{*}$
is isomorphic to  $\mathbb{Z}_{2^e} \times   \mathbb{Z}_{{p^{e_1}_1}}
 \times  \cdots \times  \mathbb{Z}_{{p^{e_k}_k}}$.
%
Let $P^+(m)$ denote the largest prime in the  prime factorization of $m$.
\begin{theorem}[Fouvry]\label{Fouvry-thm}
There exist constants $c >0$ and $n_0>0$, such that for all $x>n_0$,
\[|\{ p \mid p \mbox{ is a prime, } p < x, \mbox{ and } P^+(p-1) > p^{2/3} \}|
\ge c \frac{x}{\log x}.\]
\end{theorem}

We say  a prime $p$ satisfies the Fouvry property  if $P^+(p-1) > p^{2/3}$.
If $N=pq$, where $p$ and $q$ are distinct odd primes satisfying the
Fouvry property,
then clearly $p' =  P^+(p-1)$
appears with exponent 1 in the factorization of $p-1$,
 and so does $P^+(q-1)$ in the factorization of $q-1$. 
If $p' = P^+(p-1) > P^+(q-1)$, then $\mathbb{Z}_{p'}$
appears as an isolated  factor in the direct product form of $\mathbb{Z}_N^*$.
Thus, with probability $\ge 1- 1/p' > 1 - \frac{1}{\max\{p^{2/3}, q^{2/3}\}}
\ge 1 - N^{-1/3}$,  a random element $x$ in  $\mathbb{Z}_N^{*}$
has order at least $p'  >  \max\{p^{2/3}, q^{2/3}\}  \ge N^{1/3} $.
If it so happens that $p' = P^+(p-1) = P^+(q-1)$, then
$\mathbb{Z}_{p'} \times \mathbb{Z}_{p'}$ 
appears as a factor in the direct product form of $\mathbb{Z}_N^*$.
In this case,  a random element $x$ in  $\mathbb{Z}_N^{*}$
also has order at least $p' >  N^{1/3}$ with  probability $\ge 1- 1/(p')^2
\ge 1 - N^{-2/3}$.
Thus, in either case, in terms of the number of bits,
such products $N=pq$ have the property that a 
random element $x$ in  $\mathbb{Z}_N^{*}$ has an exponentially
large period, $\omega = \omega_N(x) \ge  \max\{P^+(p-1), P^+(q-1)\} > N^{1/3}$,
with exponentially small exceptional probability.
Below we  assume  $\omega$ has this property.

Denote by ${\rm ord}_2(x)$ the  highest power of $2$ that divides $x$.
If $e = {\rm ord}_2(p-1)$, and $e' = {\rm ord}_2(q-1)$,
then we have $2^e < (P^+(p-1))^{1/2}$
and $2^{e'} < (P^+(q-1))^{1/2}$,
and thus $\omega = \omega_N(x)$
satisfies ${\rm ord}_2(\omega) \le \max\{e, e'\} < \frac{\log_2 \omega}{2}$,
for any $x \in \mathbb{Z}_N^{*}$.
We conclude:

\begin{lemma}\label{period-large-Fouvry}
Let  $p$ and $q$  be distinct odd primes  satisfying the Fouvry property,
and let $N=pq$, then over
a random $x \in \mathbb{Z}_N^{*}$,
\[{\rm Pr.}
\left(
\omega_N(x) > N^{1/3}  ~~\mbox{\rm  and }~~ {\rm ord}_2(\omega_N(x)) 
< \frac{\log_2 \omega_N(x)}{2} \right)
> 1 - \frac{1}{N^{1/3}}.\]
\end{lemma}



\paragraph{Sum of random unit vectors}
Let $\xi_m = e^{2 \pi i/m}$ be a primitive root of unity of order $m$.
Let $X_{i} \sim N(0, 1)$, $i = 1, 2, \ldots, n$,
 be a finite sequence of i.i.d.~normally distributed random variables.
Let $\{S_k \subseteq [n] \mid 1 \le k \le K\}$ be a finite collection of
sets such that each pairwise symmetric difference $S_j \Delta S_k$
has cardinality $\ge m^2 t$, for all $j \ne k$.
Let $\Sigma_k = \sum_{i  \in S_k} X_i$ be the sum of $X_{i}$ in $S_k$.
We will give a simple estimate
for the expectation of
\begin{equation}\label{sum-unit-vectors}
 | \xi_m^{\Sigma_1} + \xi_m^{\Sigma_2} + \ldots + \xi_m^{\Sigma_K} |^2.
\end{equation}

Expanding the square norm expression we get
\[K + \sum_{1 \le j<k \le K} (\xi_m^{\Sigma_j - \Sigma_k} + \xi_m^{\Sigma_k - \Sigma_j})
= K + 2 \sum_{1 \le j< k \le K} \cos\left( (\Sigma_j - \Sigma_k) \frac{2 \pi}{m} \right).\]
Let $T_{jk} = (\Sigma_j - \Sigma_k) \frac{2 \pi}{m}$.
Note that $\Sigma_j - \Sigma_k =  \sum_{i  \in S_j \Delta S_k} (\pm  X_i)$
is a sum of at least $m^3$ distinct (thus independent) r.v.
$\pm X_i$ distributed i.i.d.~$\sim N(0, 1)$.
Therefore,
each $T_{jk}$ is a random variable normally distributed $\sim N(0, 
\sigma_{jk}^2)$, with standard deviation $\sigma_{jk}
= \sqrt{|S_j \Delta S_k|} \cdot \frac{2 \pi}{m} \ge 2 \pi \sqrt{t}$.

Moments of even orders of a normal random variable $Y\sim N(0, 
\sigma^2)$ are known as follows~\cite{Papoulis-Athanasios}
\[{\bf E}[Y^{2k}] = \sigma^{2k}(2k-1)!!,\]
from which we get  (by the dominated convergence theorem, the exchange of
orders of summation and integration is justified)
\begin{eqnarray*}
 {\bf E}[ \cos (T_{jk})] &=& 1 - \frac{\sigma_{jk}^2}{2!} (2-1)!!
+ \frac{\sigma_{jk}^4}{4!} (4-1)!! - \frac{\sigma_{jk}^6}{6!} (6-1)!! + \ldots\\
&=& e^{- \sigma_{jk}^2/2} \\
&\le& e^{- 2 \pi^2 t}.
\end{eqnarray*}
Hence, the expectation of (\ref{sum-unit-vectors}) is at most
$K + 2 {K \choose 2}e^{- 2 \pi^2 t}$.


We will need a slight generalization of this.
Let $\sigma >0$, and let $\varphi_k \in [0, 2 \pi)$ be any angle, $1 \le k \le K$.
We replace each $\Sigma_k$ by $\varphi_k + \sigma \sum_{i \in S_k} X_{i}$.
Then,
\begin{lemma}\label{lemma-exp-sq-norm-sum}
Let $\sigma >0$ and 
$\xi_m = e^{2 \pi i/m}$.
Let $X_{i} \sim N(0, 1)$, i.i.d. for $i = 1, 2, \ldots, n$,
and let $\{S_k \subseteq [n] \mid 1 \le k \le K\}$ be a finite collection of
sets. Assume, all except   at most $\delta$ fraction
of  pairwise symmetric differences $S_j \Delta S_k$ have
cardinality $\ge (m/\sigma)^2 t$ for $j \ne k$. 
Let
 $\Sigma_k = \varphi_k +  \sigma \sum_{i  \in S_k} X_i$, where  $\varphi_k \in [0, 2 \pi)$.
Then,
 \[  {\bf E}[| \xi_m^{\Sigma_1} + \xi_m^{\Sigma_2} + \ldots + \xi_m^{\Sigma_K} |^2]
\le K +  2 \delta  {K \choose 2} + 2 (1-\delta)  {K \choose 2} e^{- 2 \pi^2  t}.\]
\end{lemma}
\begin{proof}
Let $T_{jk} = \frac{2 \pi \sigma}{m}(\sum_{i \in S_j} X_i - \sum_{i \in S_k} X_i)$.
We only need to note in addition to the above  that
\[\cos \left( \varphi + T_{jk}  \right)
= \cos  \varphi \cos T_{jk} - \sin  \varphi \sin  T_{jk},\]
and we have $\cos  \varphi \le 1$ for any $\varphi$,
and ${\bf E}[ \sin  T_{jk}] =0$ since $\sin$ is an odd function
and $T_{jk}$ is symmetrically distributed. The lemma follows.
\end{proof}

\section{Corrupted geometric sums}
Suppose $N$ is an integer we wish to factor, and $2^n \approx N^2$ as in~\cite{shor}~\footnote{Thus $N$
has $\approx n/2$ bits, a slight change in notation from Section~\ref{sec:introduction}.}.
For definiteness assume $2^{n-1} < N^2 \le 2^{n}$.
Assume $\omega$ is the period of the function $f(k) = x^k \bmod N$
for a randomly chosen $x \in \mathbb{Z}_N^*$, and by Lemma~\ref{period-large-Fouvry}
we assume $ \omega > N^{1/3}$
and
 ${\rm ord}_2(\omega) < \frac{\log_2 \omega}{2}$. 
Also $\omega < N$ clearly.

Let us write out a few terms as the controlled-$R_k$ gates are applied successively
in the  QFT circuit (e.g., see~\cite{book} p.219), but now
with random noise added whenever the controlled-rotation gate is $R_k$-gates
with $k \ge b$, i.e., we apply controlled-$R_k$-gates when $k<b$ but
controlled-$\widetilde{R_k}$-gates
 for all $k \ge b$. (As the first controlled-$R_k$-gate has $k=2$,
we have $b >1$.)
Suppose we start with the state $|u\rangle = |u_{n-1} \ldots u_1u_0\rangle$.
After the first gate $H$ on the qubit $|u_{n-1}\rangle$, we have the state
\[\frac{1}{2^{1/2}}
\left(|0\rangle + e^{2 \pi i~0.u_{n-1}}|1\rangle \right) |u_{n-2} \ldots u_0\rangle.\]
The next is the controlled-$R_2$-gate on target qubit $|u_{n-2}\rangle$ controlled by
the left most qubit (which was initially $|u_{n-1}\rangle$), after which we have (assuming $b>2$) 
\[\frac{1}{2^{1/2}}
\left(|0\rangle + e^{2 \pi i~0.u_{n-1}u_{n-2}}|1\rangle \right) |u_{n-2} \ldots u_0\rangle.\]
The random noise starts at the controlled-$R_b$-gate, after which we get
\[\frac{1}{2^{1/2}}
\left(|0\rangle + e^{2 \pi i~\left[0.u_{n-1}\ldots u_{n-b} + \frac{\epsilon}{2^b}
u_{n-b} r^{(0)}_0\right]}|1\rangle \right) |u_{n-2} \ldots u_0\rangle,\]
where $r^{(0)}_0 \sim N(0,1)$.

After all the rotation gates controlled by the  left most qubit (initially
$|u_{n-1}\rangle$) we have
\begin{equation}\label{first-sweep-eqn-second-model}
\frac{1}{2^{1/2}}
\left(|0\rangle + e^{2 \pi i~\left[0.u_{n-1}\ldots u_{0} + \frac{\epsilon}{2^b}
\left( u_{n-b}r^{(0)}_0 + \frac{u_{n-b-1}r^{(0)}_1}{2} + \cdots + \frac{u_{0}r^{(0)}_{n-b}}{2^{n-b}}\right)\right]}|1\rangle \right) |u_{n-2} \ldots u_0\rangle,
\end{equation}
%
where $r^{(0)}_0, \ldots, r^{(0)}_{n-b}$ are i.i.d.~$\sim N(0,1)$.

Then,
similarly,
after all the rotation gates controlled by the  two left most qubits
 (initially
$|u_{n-1}u_{n-2}\rangle$)  we have
{\tiny
\begin{equation}\label{second-sweep-eqn}
\frac{1}{2^{2/2}}
\left(|0\rangle + e^{2 \pi i~\left[0.u_{n-1}\ldots u_{0} + \frac{\epsilon}{2^b}
\left( u_{n-b}r^{(0)}_0
+ \cdots + \frac{u_{0}r^{(0)}_{n-b}}{2^{n-b}}\right)\right]}|1\rangle \right)
\otimes
\left(|0\rangle + e^{2 \pi i~\left[0.u_{n-2}\ldots u_{0} + \frac{\epsilon}{2^b}
\left( u_{n-b-1}r^{(1)}_0
+ \cdots + \frac{u_{0}r^{(1)}_{n-b-1}}{2^{n-b-1}}\right)\right]}|1\rangle \right)
|u_{n-3} \ldots u_0\rangle,
\end{equation}
}
where $r^{(0)}_0, \ldots, r^{(0)}_{n-b}, r^{(1)}_0,  \ldots, r^{(1)}_{n-b-1}$ are i.i.d.~$\sim N(0,1)$.


The circuit continues to apply controlled rotation gates with random noise 
starting
at the controlled-$R_b$-gate, producing a final expression with $n$ tensor factors. When written
out the tensor product, this  is
a sum 
indexed by  $|v_{n-1} \ldots v_0\rangle$, such that
$v_0 = 0$ or 1 corresponds to selecting respectively the term $|0\rangle$ or
$e^{2 \pi i~\left[\cdots\right]} |1\rangle$ in (\ref{first-sweep-eqn-second-model})
(or equivalently, to selecting one of the two terms in
the first tensor factor in (\ref{second-sweep-eqn})), 
and $v_1 = 0$ or 1 corresponds to selecting respectively the term $|0\rangle$ or
$e^{2 \pi i~\left[\cdots\right]} |1\rangle$ in the second tensor factor in 
 (\ref{second-sweep-eqn}), 
and similarly for  $v_s = 0$ or $1$, for all $0 \le s \le n-1$.

The crucial step in Shor's algorithm,
after the quantum Fourier transform,
is to take a quantum measurement, with the property
that  the probability of observing a state that is close to
an integral multiple of $\frac{2^n}{\omega}$ is high.
Such a state has an $n$-bit integer expression $v \in \{0, 1\}^n$
that has value close to the rational number $\frac{2^n}{\omega} j$,
for some $0 \le j \le \omega$.
States $|v\rangle$ such that the number $v$
is not close to an integral multiple of $\frac{2^n}{\omega}$
have negligible probability of being observed, while states 
in a small vicinity of each of the  integral multiple of $\frac{2^n}{\omega}$
get observed with probability on the order of $1/\omega$ (per each multiple),
and these add up to give a good probability that some such
state is observed, whereby the period is deduced with good probability.
(I am omitting steps of the continued fraction algorithm in the
post quantum processing steps.)

For each $v$,
the probability of $|v \rangle$ being observed
has an expression as a square norm
of a sum over a set of the form $u \in \{u^* + k \omega: k \ge 0, \mbox{ and }
u^* + k \omega < 2^n\}$ (for some initial $0 \le u^* < \omega$), 
with cardinality $K$, which is approximately $2^n/\omega$. 
For $u^{(k)} = u^* + k \omega$, 
we write the $n$-bit integers $u^{(k)} = \sum_{s=0}^{n-1} u^{(k)}_s 2^{s}$
and $v  = \sum_{s=0}^{n-1} v_s 2^{s}$.
When there is no noise in the controlled-$R_k$-gates used in the QFT
this probability expression for observing $|v \rangle = | v_{n-1} \ldots v_1 v_0\rangle$ can be written as
\[
\frac{1}{2^n K}\left| \sum_{k=0}^{K-1}  
\exp\left\{ 2 \pi i \sum_{t=1}^n \frac{\sum_{s=0}^{n-t} u^{(k)}_{n-t-s} v_{s}}{2^{t}}\right\}\right|^2.\]

With independent  random  noise present starting with controlled-$R_b$-gates,
this 
becomes
{\tiny
\begin{equation}\label{currepted-prob-second-model}
\frac{1}{2^n K}\left| \sum_{k=0}^{K-1}
\exp\left\{ 2 \pi i \left[
 \sum_{t=1}^n \frac{\sum_{s=0}^{n-t} u^{(k)}_{n-t-s} v_{s}}{2^{t}}
+ \frac{\epsilon}{2^b} \left\{ \left( u^{(k)}_{n-b}r^{(0)}_0
+ \cdots + \frac{u^{(k)}_{0}r^{(0)}_{n-b}}{2^{n-b}}\right) v_0 +
\left( u^{(k)}_{n-b-1}r^{(1)}_0
+ \cdots + \frac{u^{(k)}_{0}r^{(1)}_{n-b-1}}{2^{n-b-1}}  \right) v_1
+ \ldots + u^{(k)}_{0}r^{(n-b)}_0 v_{n-b}  \right\} \right]
\right\}  \right|^2,
\end{equation}
}
 where
 \[r^{(0)}_0, \ldots, r^{(0)}_{n-b}, r^{(1)}_0, \ldots, r^{(1)}_{n-b-1},\ldots,
r^{(n-b-1)}_0, r^{(n-b-1)}_1,
r^{(n-b)}_0\]
 are random variables i.i.d.~$\sim N(0,1)$.


Our first goal is to show that among states $|v \rangle$
such that the binary number $v$ is close to an integral multiple
 $\frac{2^n}{\omega} j$
(for some $0 \le j \le \omega$), it is the case that  for most $j$,
a linear number of  bits in the 
binary expansion of $v$ are one: 
 $v_s = 1$. This will leave us with a  linear number of  
terms of the form  in the exponent
\[\frac{2 \pi i \epsilon}{2^b} \left( u^{(k)}_{n-b-s}r^{(s)}_0
+ \frac{u^{(k)}_{n-b-s-1}r^{(s)}_1}{2}
+ \cdots + \frac{u^{(k)}_{0}r^{(s)}_{n-b-s}}{2^{n-b-s}}  \right) v_s.\]
Eventually we will show that,
fixing any such $v$,
among those $s$ where $v_s =1$,
 for most $k$, there are a  linear number of
terms with $u^{(k)}_{n-b-s} =1$, which will give us the perturbation
as a sum of $\frac{2 \pi i \epsilon}{2^b} \cdot r^{(s)}_0$.

%
Let us consider integers
$v= \lfloor \frac{2^n}{\omega} j \rfloor$, for $0 \le j < \omega$;
it will be clear from the proof below that what is proved  is also true for
any $v$ in the vicinity of a polynomial range of such a number.

For $0 \le j  < \omega$, the integer $v =\lfloor \frac{2^n}{\omega} j \rfloor$
has the $i$-th leading bit $v_{n-i}= 1$ iff
the $i$-th most significant bit,
 among the first $n$ bits, in the binary expansion of
$\frac{j}{\omega}$ is  1. This is true iff
for some $1 \le k \le 2^{i-1}$,
\[
\frac{2k-1}{2^i} \le \frac{j}{\omega} < \frac{2k}{2^i},\]
which is equivalent to
\begin{equation}\label{odd-index-for-i}
(2k-1)\frac{\omega}{2^i} \le j < 2k \frac{\omega}{2^i}.
\end{equation}
So, $j$ needs to be placed in the alternate (``odd'' indexed) segments
of length $\frac{\omega}{2^{i}}$. This is a lattice counting problem.

Recall that $\omega > N^{1/3} \approx 2^{n/6}$.
We take $i_0 = \lfloor \frac{3}{4} \log_2  \omega \rfloor \ge 
 \lfloor \frac{1}{4}  \log_2 N \rfloor = {\Omega(n)}$.
Then $\frac{\omega}{2^{i_0}} \ge \omega^{1/4}> N^{1/12} = 2^{\Omega(n)}$.
We will only count those $i$-th (significant) bits $v_{n-i}$ that are one, within $1 \le i \le i_0$,
and first show that for most $j$, even just among the first $i_0$
bits $v_{n-1}, \ldots, v_{n-i_0}$, there are a linear number of ones.
(Any additional bits that are 1 can only add more noise to the perturbation.)

Now we divide the range
$[0, \omega)$ of real numbers  into $2^{i_0}$ segments 
of equal length $\frac{\omega}{2^{i_0}}$ 
\[ I_\alpha = \{ x \in \mathbb{R} \mid \frac{\omega}{2^{i_0}}(\alpha)_2  \le x < 
\frac{\omega}{2^{i_0}}((\alpha)_2 +1)  \},\]
where $\alpha \in \{0, 1\}^{i_0}$ is a binary string, and $(\alpha)_2$
is the binary number it represents~\footnote{The reason we cut off at $i_0$ is
to avoid having to deal with intervals that are too small  and such odd indexed
segments may just miss most integers. We can afford to cut off at  $i_0$, and still
get a linear number $\Omega(n)$ of 1's in the first $i_0$ bits of
the $n$ bit binary expansion. This is where
we use the fact that
$\omega$ is large.}.

Note that any real interval of the form $[A, A+B)$ has 
either
$\lfloor B \rfloor $ or $\lfloor B \rfloor +1 $ many integers.
Thus, each $I_\alpha$ contains either $\lfloor \frac{\omega}{2^{i_0}} \rfloor$
or $\lfloor \frac{\omega}{2^{i_0}} \rfloor +1$ 
many integers, which is $\frac{\omega}{2^{i_0}} + \eta$
for some $-1 \le \eta \le 1$.
We consider two distributions on the integers $0 \le j < \omega$.
Let ${\rm Pr.}$ denote the uniform distribution and
let ${\rm Pr}_\alpha$ denote the distribution induced by
first picking  $\alpha\in \{0, 1\}^{i_0}$ uniformly, and then 
picking $j \in I_{\alpha}$ uniformly.
They are exponentially close: For any $0 \le j < \omega$, ${\rm Pr.}(j) = 1/\omega$, and
\begin{equation}\label{close-two-distributions}
{\rm Pr}_\alpha(j) = \frac{1}{2^{i_0}} \frac{1}{\frac{\omega}{2^{i_0}} +\eta} =
\frac{1}{\omega + \eta  2^{i_0}}
=  \frac{1}{\omega} \cdot \frac{1}{1 + \eta \frac{2^{i_0}}{\omega}} 
= \left(1 \pm 2^{-\Omega(n)} \right) \cdot {\rm Pr.}(j) .
\end{equation}

Let $\alpha = \alpha_1 \alpha_2 \ldots \alpha_{i_0}$.
Consider any $j \in I_\alpha$. 
If $\alpha_{i_0} =1$, then $j$ satisfies the equation (\ref{odd-index-for-i})
for $i=i_0$.
%
Now suppose $\alpha_{i_0-1}=1$, then
\[ \frac{\omega}{2^{i_0}} (\alpha_1 \ldots \alpha_{i_0-1} \alpha_{i_0})_2
\ge \frac{\omega}{2^{i_0}} (\alpha_1 \ldots \alpha_{i_0-1}  0)_2
= \frac{\omega}{2^{i_0-1}} (\alpha_1 \ldots \alpha_{i_0-1})_2 \]
and
\[\frac{\omega}{2^{i_0}} ((\alpha_1 \ldots \alpha_{i_0-1} \alpha_{i_0})_2 + 1)
\le \frac{\omega}{2^{i_0-1}}  \frac{(\alpha_1 \ldots \alpha_{i_0-1} 0)_2 + 2}{2}
= \frac{\omega}{2^{i_0-1}} ((\alpha_1 \ldots \alpha_{i_0-1})_2 + 1).
 \]
And so clearly $j$ satisfies the equation (\ref{odd-index-for-i})
for $i=i_0-1$.

Similarly, we can see that every $j \in I_\alpha$ satisfies the equation (\ref{odd-index-for-i})
for every $i \in \{1, \ldots, i_0\}$ such that the corresponding bit
in $\alpha$ is 1. For any constant $0< \delta < 1/2$,
the proportion of 0-1 sequences of length $i_0$ that have
$\delta i_0$ ones is asymptotically $2^{-(1- H(\delta))i_0}$,
where $H(\cdot)$ is the entropy function.  
For any fixed constant $c>0$, consider any 
$J = \{i: i_0' \le i \le i_0 \}$ with  length $i_0 - i_0' +1 \ge cn$
indexing bit positions  $\alpha_{i_0'}, \ldots, \alpha_{i_0}$.
Then, for a random $\alpha \in \{0, 1\}^{i_0}$, with exponentially small exceptional probability $2^{-\Omega(n)}$,
there are
 $\Omega(n)$ bits $\alpha_i=1$  in those bit positions $i \in J$.
Then any $j \in I_\alpha$ gives the corresponding bit $v_{n-i} =1$.
By (\ref{close-two-distributions}) this is true under the uniform
distribution ${\rm Pr.}$ for $j$ as well.
It follows that with exponentially small exceptional probability $2^{-\Omega(n)}$,
 a uniformly chosen $j$ 
defines a number $v= \lfloor \frac{2^n}{\omega} j \rfloor$  with a
linear number of bit  $v_{n-i} = 1$, for $i \in J$.
\begin{lemma}\label{lm:most-v-are-good}
For any fixed constant $c>0$ and any 
$J = \{i: i_0' \le i \le i_0 \}$
 with  length $i_0 - i_0' +1 \ge cn$,
picking a random $0 \le j < \omega$ uniformly which defines $v= \lfloor \frac{2^n}{\omega} j
\rfloor$,
\[{\rm Pr.} \left( \left|\{i \in J : v_{n-i} =1\}\right| \ge \Omega(n) \right)
= 1 - 2^{-\Omega(n)}.\]
\end{lemma}

Now back to the expression (\ref{currepted-prob-second-model})
 for the probability of observing
$|v\rangle$ when noise is present.
Regardless what values \[\sum_{t=1}^n \frac{\sum_{s=0}^{n-t} u^{(k)}_{n-t-s} v_{s}}{2^{t}},~~~\mbox{and}~~~r^{(0)}_1, \ldots, r^{(0)}_{n-b}, r^{(1)}_1, \ldots, r^{(1)}_{n-b-1},\ldots,
r^{(n-b-1)}_1\]
are,
let us consider only those terms
\begin{equation}\label{eqn:mainsum-for-r_0}
\frac{2 \pi  \epsilon}{2^b}
\left(u^{(k)}_{n-b}v_0 r^{(0)}_0 + u^{(k)}_{n-b-1}v_1 r^{(1)}_0
+ \ldots + u^{(k)}_0 v_{n-b} r^{(n-b)}_0\right)
=
\frac{2 \pi  \epsilon}{2^b} \sum_{i=b}^{n} u^{(k)}_{i-b} v_{n-i} r^{(n-i)}_0.
\end{equation}

We will further throw away some noise terms in (\ref{eqn:mainsum-for-r_0}).
Let $d={\rm ord}_2(\omega)$.
Recall that  $d < \frac{\log_2 \omega}{2}$
and $i_0 = \lfloor \frac{3}{4} \log_2  \omega \rfloor$.
Thus, assuming $b$ is $O(\log n)$,
$i_0 - b-d = \Omega(n)$, and we will only consider
the subsum in (\ref{eqn:mainsum-for-r_0})
for $i \in \{d+b, \ldots, i_0\}$, which has
$\Omega(n)$ terms.

By Lemma~\ref{eqn:mainsum-for-r_0},
except for an exponentially small fraction $2^{-\Omega(n)}$ of $j$
indexing $v =  \lfloor \frac{2^n}{\omega} j \rfloor$  ($1 \le j < \omega$),
 each $j$ defines a linear sized $T_j = \{ d+b \le i \le i_0 :
v_{n-i} = 1\}$ (of cardinality $ > 
\Omega(n)$)
such that $v_{n-i} = 1$ and
so $u_{i-b} v_{n-i} r^{(n-i)}_0 = u_{i-b} r^{(n-i)}_0$, for $i \in T_j$.
Thus we will 
further ignore a large portion of the above sum
(\ref{eqn:mainsum-for-r_0}), and
consider only
\begin{equation}\label{subsum-on-good-v}
\frac{2 \pi  \epsilon}{2^b}
\sum_{i \in T_j}
u^{(k)}_{i-b} r^{(n-i)}_0.
\end{equation}
Intuitively,  any term that in fact survives (i.e., with $u^{(k)}_{i-b}v_{n-i} =1$)
that we omitted can only increase the noise. (Formally,
 when we eventually apply
Lemma~\ref{lemma-exp-sq-norm-sum}, these will all be part of
the term $\varphi_k$.)

Our next goal is to show that, among $i \in T_j$,
most pairs of $u^{(k)} = u^* + k \omega$
and $u^{(k')} = u^* + k' \omega$, for $ k \ne k'$, 
 have a linear number of different bit values
$u^{(k)}_{i-b} \ne u^{(k')}_{i-b}$, for $i \in T_j$.



To investigate the (least $i_0 - b +1$ significant)  
bits $u^{(k)}_{0}, u^{(k)}_1, \ldots, u^{(k)}_{i_0-b}$
of $u^{(k)} = u^* + k \omega$, 
we consider $u^{(k)} \bmod 2^{i_0-b+1}$.
If $\omega$ is odd, then $(k \omega \bmod 2^{i_0-b+1})$ will enumerate
all values in  $\{0, 1, \ldots, 2^{i_0-b+1} -1\}$ exactly once,
when $k = 0, 1,  \ldots, 2^{i_0-b+1} -1$.
Our range of $k$ is actually from 0 to just below $\frac{2^n - u^*}{\omega}
\approx 2^n/\omega
\gg 2^{i_0}$.
Thus, for any $u^*$,  $(u^* + k \omega \bmod 2^{i_0-b+1})$
enumerates every value in  $\{0, 1, \ldots, 2^{i_0-b+1} -1\}$
 almost uniformly.

In general,
$0 \le d={\rm ord}_2(\omega) < \frac{\log_2 \omega}{2}$, and thus
$i_0 - b-d = \Omega(n)$ for $b = O(\log n)$,
 and then
 $\omega' = \omega/2^d$ is invertible in  $\mathbb{Z}_{2^{i_0-b+1-d}}$,
and for any $u^*$,  the most significant $i_0-b+1-d = \Omega(n)$
bits in $(u^{(k)} \bmod 2^{i_0-b+1})$  are almost uniform.
These are the  most significant $i_0-b+1-d$ of the
$i_0-b+1$ least significant bits of $u^{(k)}= u^* + k \omega$.

Consider $u^{(k)} = u^* + k \omega$ and $u^{(k')} = u^* + k' \omega
= u^{(k)} + (k'-k) \omega$.
For any $k$, let $k'$ run through $\{0, \ldots, 
\lfloor \frac{2^n - u^*}{\omega}
\rfloor\}$, then $k'-k$ runs through  $\{-k, \ldots, 
\lfloor \frac{2^n - u^*}{\omega}
\rfloor -k\} $, a set of consecutive integers of 
size $\ge \lfloor \frac{2^n}{\omega} \rfloor
\ge 2^{i_0-b+1-d} \left(\frac{2^n}{2^{i_0} \omega}\right)$.
As $2^{i_0} \omega \le 2^{7n/8}$ we have $\frac{2^n}{2^{i_0} \omega}
\ge 2^{n/8}$. Hence, $(k'-k)\omega' \bmod 2^{i_0-b+1-d}$ picks
every value in $\{0, \ldots, 2^{i_0-b+1-d}-1\}$
with probability $\frac{1}{2^{i_0-b+1-d}} \cdot ( 1 \pm 2^{-\Omega(n)})$.
Thus,
the  most significant $i_0-b+1-d =\Omega(n)$ of the
$i_0-b+1$ least significant bits of $u^{(k')} - u^{(k)}=  (k'-k) \omega$
are almost uniform, as $k'$ runs through $\{0, \ldots,
\lfloor \frac{2^n - u^*}{\omega}
\rfloor\}$.
Then using the same argument with the entropy function $H (\cdot)$,
for all except a fraction of $2^{-\Omega(n)}$ of the pairs $(k,k')$,
we have  $u^{(k)}_{i-b} \ne u^{(k')}_{i-b}$, for a subset of $i \in T_j$
of cardinality $\ge c_0 n$, where the constant $c_0 >0$ is uniform for $(k,k')$.
\begin{lemma}\label{u-are-mostly-good}
Assume $|T_j| = \Omega(n)$. There exists $c_0 >0$, such that
for random pairs $(k,k')$,
\[{\rm Pr.} \left(\left| \{i \in T_j: u^{(k)}_{i-b} \ne u^{(k')}_{i-b} \} \right| \ge
c_0 n \right) = 1 - 2^{-\Omega(n)}.\]
\end{lemma}


It follows that, except for a $2^{-\Omega(n)}$ fraction of pairs 
$(k,k')$, the sum
\[ \sum_{i=b}^{n} (u^{(k)}_{i-b} - u^{(k')}_{i-b}) v_{n-i} r^{(n-i)}_0\]
contains a linear number $n' \ge c_0 n$ of uncancelled terms $r^{(n-i)}_0$
where $v_{n-i}=1$ and $u^{(k)}_{i-b}  \not = u^{(k')}_{i-b}$.
To apply  Lemma~\ref{lemma-exp-sq-norm-sum},
we require $(\frac{\epsilon}{2^b})^{-1} < ({n'})^{1/3}$,
or equivalently $b + \log 1/{\epsilon} < \frac{1}{3} \log n'$.

To summarize the error estimates:
(I) except with probability $2^{-\Omega(n)}$,
 we have $\omega > N^{1/3}$ and  ${\rm ord}_2(\omega)
< \frac{\log_2 \omega}{2}$ by Lemma~\ref{period-large-Fouvry};
(II) except for a fraction of  $2^{-\Omega(n)}$ of $j$'s, 
all $v = \lfloor \frac{2^n}{\omega} j \rfloor$ have $|T_j| = \Omega(n)$
by Lemma~\ref{lm:most-v-are-good};
(III) except for a fraction of  $2^{-\Omega(n)}$ of all pairs $(k,k')$'s,
the sums (\ref{subsum-on-good-v}) defined by $k$ and $k'$ all
have a symmetric difference 
 with cardinality $ \ge  (2^b/\epsilon)^3$,
by Lemma~\ref{u-are-mostly-good}.

Finally we estimate the sum  of 
 the expectations of the square norm sum
(\ref{currepted-prob-second-model}) indexed by 
all $v= \lfloor \frac{2^n}{\omega} j \rfloor$.
Note that the 
sum $\sum_{k=0}^{K-1}$ is over $K$ complex numbers of
unit norm, and thus has norm at most $K$.
With probability $\le 2^{-\Omega(n)}$, (I) may be
violated and the sum over all $v = \lfloor \frac{2^n}{\omega} j \rfloor$ 
of (\ref{currepted-prob-second-model})
can be at most $ \frac{\omega}{2^nK} K^2 = O(1)$.
Assume (I) holds, then the sum of the terms (\ref{currepted-prob-second-model})
indexed by
the $\le 2^{-\Omega(n)}$ fraction of  exceptional
$v$'s regarding (II) has value at most $(2^{-\Omega(n)} \omega) \frac{1}{2^nK} K^2
= 2^{-\Omega(n)}$.
Assume (I) and (II) are both not violated, 
we apply Lemma~\ref{lemma-exp-sq-norm-sum}. 
By (III), we get the estimate over all $v$, the expectation
$\frac{\omega}{2^n K} \left(K + 2^{-\Omega(n)} K^2 + K^2 2^{-\Omega(n^{1/3})} \right)
= 2^{-\Omega(n^{1/3})}$.

We conclude that 
the expectation (over the random noise bits $r$'s) of the 
probability of observing a member in $\{|v \rangle : 
v= \lfloor \frac{2^n}{\omega} j \rfloor, 0 \le j < \omega\}$ 
is exponentially small.



The  proof carries over easily to  those $|v \rangle$ that are 
 in the vicinity of a polynomial range of $\lfloor \frac{2^n}{\omega} j \rfloor$.
And since the  estimate is exponentially small,  the proof shows that
the  probability of observing any member of the set of
those $|v \rangle$ that are polynomially close
to an integral multiple $\lfloor \frac{2^n}{\omega} j \rfloor$
is still exponentially small in expectation.

\section{Some comments}\label{Comments.tex}
This section contains some comments and personal opinions. They are
speculative, and are not to be conflated with the provable part. 

Quantum mechanics is unquestionably an accurate model of 
microscopic physical reality. However, I believe every physical
theory is an approximate description of the  real world,
and quantum mechanics is no exception. In particular, I believe the
${\rm SU}(2)$ description of possible operations of a qubit 
to be only approximately true.  Specifically, I don't believe arbitrarily
small angles have physical meaning.

The real numbers $\mathbb{R}$, the continuum, is a human logical
construct in terms of Dedekind cut or Cauchy sequence in the 
language of $\epsilon$-$\delta$
definition. ${\rm SU}(2)$ (or equivalently ${\rm SO}(3)$) as a group,
is built on top of the continuum.
That these mathematical objects provide remarkable fit in some
mathematical theory for physical reality, is an extraordinary  fact.
But this extraordinary fit is always within a certain range; its
unlimited extrapolation is mathematical idealization.
The Schr\"{o}dinger equation
$i \hbar \frac{d}{dt} |\Psi(t)\rangle  = \hat{H} |\Psi(t)\rangle$
suggests that small angles are related to small time periods.
But physicists have
suggested that time ultimately also comes in discrete and indivisible
``units''. The concept   ``chronon'' has been
proposed as  a quantum of time~\cite{Margenau}.
It has even been proposed that
 one chronon corresponds to about  $6.27 \times 10^{-24}$ 
seconds for an electron, 
 much longer than the Planck time, which is only about $5.39 \times 10^{-44}$~\cite{Caldirola} (see also~\cite{Farias,Albanese}).
%
%
(Of course the literal form of the mathematical meaning of Schr\"{o}dinger equation,
as a differential
equation, suggests time is infinitely divisible. But my personal view
is that this is just mathematical abstraction.) 

Thus, I view arbitrarily small angles permitted under ${\rm SU}(2)$
as mere mathematical abstraction.  It is true that using a fixed finite
set of rotations of reasonable angles such as $\pi/8$ along various axes
\emph{can} compose to rotations of arbitrarily small angles. But my view
is that these compositional rules as specified by the group  ${\rm SU}(2)$
must not be exact for physical reality. And thus, it seems
to me permitting
some noise in the model is reasonable. The random noise model
in this paper is just a model, is not meant as reality.

Of course, in addition to  its intrinsic interest, factoring
integers of the form $pq$ is at the heart of the RSA public-key
cryptosystem~\cite{RSA}. But several results and conjectures
in number theory suggest that the failure reported in
this paper of Shor's factoring algorithm
in the presence of noise can be more severe
in the asymptotic sense.
  We used a theorem of Fouvry~\cite{Fouvry} to
produce a set of primes of positive density that have the desired properties
of the period of a random element. The most important property
is that this period is sufficiently large. 
In Theorem~\ref{main3} we prove a version of the theorem for primes 
of density one. There are deep
results  and many conjectures
about the distribution of prime factorizations of $p-1$.
In the extreme there are the so-called Sophie Germain primes $p'$
such that $p = 2 p' +1$ is also a prime. It is conjectured that
there are $2C \frac{x}{(\log_e x)^2}$ many
Sophie Germain primes up to $x$,
where $C = \prod_{p>2}\frac{p(p-2)}{(p-1)^2} \approx 0.660161$
is the Hardy-Littlewood twin prime constant. 
 This is just slightly less than positive density.
(However, it has not been proved that there are 
 infinitely many Sophie Germain primes.)
%
Sophie Germain primes were studied in (the first case of)
Fermat’s Last Theorem. Indeed, Adleman 
and Heath-Brown~\cite{Adleman-Heath-Brown}, and Fouvry~\cite{Fouvry} 
 proved  that the
first case of Fermat's Last Theorem holds for infinitely many primes $p$.
(See also~\cite{Lenstra-Stevenhagen}.)  
In~\cite{Hastad-Schrift-Shamir}, H\r{a}stad, Schrift and Shamir
(acknowledging Noga Alon)
proved that, for any fixed constant $k$,
  for randomly chosen primes $p$ and $q$ of equal size,
and $N=pq$ of size $n$, 
a random element in $\mathbb{Z}_N^*$
 has order $\ge \phi(pq)/n^k$ except with  probability $O(n^{-(k-5)/5})$.
We improve this slightly in Theorem~\ref{HSS-period-poly} 
in the process of proving
Theorem~\ref{main3}.

Another property  we use of primes of the property of
Theorem~\ref{Fouvry-thm} is that the period of a random element
in  $\mathbb{Z}_N^*$ does not have high ${\rm ord}_2$.
Erd\"{o}s and Odlyzko~\cite{Erdos-Odlyzko}
proved that the set of odd divisors of
$p-1$ 
has a positive density.

Finally, a few comments on the Strong Church-Turing thesis.
It is conceivable that some other quantum algorithm in the BQP model can 
factor integers (or some other seemingly difficult problem)
 in polynomial time, \emph{and} withstand the random noise
discussed in this paper. Separately, it is definitely conceivable
that at some future time,
a quantum algorithm is superior to the best ``classical'' factoring algorithms
for integers of a certain range.
But I am not convinced that BQP requires that we revise
the Strong Church-Turing thesis, even if factoring
is eventually known to be outside P or BPP.  In Turing's careful definition
 of computability, he made a deliberate choice that the ``primitive''
steps of such a computing device must be discrete. Thus,
the set of states of a TM is finite; the symbols are placed
in discrete cells; the alphabet set is finite.  At its most
fundamental level, it is not permitted
to ask the  computing machine to scan a continuously deformed
symbol from $\xi$ to $\zeta$, while a mathematical homotopy
can easily be envisioned. I believe the model BQP, in its
use of the full  ${\rm SU}(2)$ as primitive steps
 (or what amounts to equivalently, the assumption
that the exact rule of composition of  ${\rm SU}(2)$
corresponds exactly to realizable computational steps), is a departure from the Turing model. 

\section*{Acknowledgement}
The author thanks
Al Aho,
Dan Boneh,
P\'{e}ter G\'{a}cs,
Zvi Galil,
Fred Green,
Steve Homer,
Leonid Levin,
Dick Lipton,
Ashwin Maran,
Albert Meyer,
Ken Regan,
Ron Rivest,
Peter Shor,
Mike Sipser,
Les Valiant, and
Ben Young
for insightful comments.
He also thanks
Eric Bach for inspiring discussions
on some of the number theoretic estimates, and we hope to
report some further improvements soon~\cite{B-C}.
A similar result can be proved for Shor's algorithm computing
 Discrete Logarithm,
and will be reported later.

\section*{Appendix: Pairs of random primes}\label{random-primes}
The proof in the paper exhibits a particular set of
primes of positive density, and shows that if the input $N$
to Shor's algorithm is of the form $N=pq$ for any primes $p$ and $q$
from that set then the algorithm does not factor with exponentially
small exceptional probability, \emph{if} the rotational gates are
accompanied by a suitable level of noise.

In cryptography, an interesting question concerns the performance
on  $N=pq$ for random primes $p$ and $q$ of length  $m$.  

In this appendix, we prove Theorem~\ref{main3}, 
dealing with  random pairs of primes  $p$ and $q$
chosen uniformly from all primes of the same length.

%
%
To prove Theorem~\ref{main3}, we will appeal to some number theoretic estimates
for 
\begin{itemize}
\item
The period $\omega_N(x)$ of a random element $x \in \mathbb{Z}_{N}^*$,
where $N =  pq$, and $p$ and $q$ are primes
uniformly randomly  chosen  from all primes of length $m$.
(The period $\omega_N(x)$ is the order of $x$ as a group element
in  $\mathbb{Z}_{N}^*$.) 
\item
The exact order of the prime 2 of the integer $\omega_N(x)$, 
i.e., ${\rm ord}_2(\omega_N(x))$,
for a random element $x \in \mathbb{Z}_{N}^*$,
where $N =  pq$,
 and $p$ and $q$ are primes
uniformly randomly  chosen  from all primes of length $m$.
\end{itemize}

For  primes $p$ and $q$ of binary length $m$,
$N=pq$ has binary length $\approx 2m$, and the QFT circuit
uses about $4m$ qubits with $2^{4m} \approx  N^2$. 
The statement  $b+ \log 1/\epsilon < \frac{1}{3} \log m -c$ for some $c>0$
is equivalent to $b+ \log 1/\epsilon < \frac{1}{3} \log (4m) -c'$ for some $c'>0$.
We note that to carry through the same proof of the Main theorem
in the paper, we only need to have the property that
\begin{enumerate}
\item
 $\omega_N(x)$ is large, say $\omega_N(x) = 2^{\Omega(m)}$,
 and
\item
 ${\rm ord}_2(\omega_N(x))$
is not too large, say ${\rm ord}_2(\omega_N(x)) = o(m)$.
\end{enumerate}

H\r{a}stad, Schrift and Shamir proved
a version of the following theorem (Theorem~\ref{HSS-period})
(acknowledging Noga Alon)~\cite{Hastad-Schrift-Shamir}
 [Proposition 1, p.~378]. Their theorem is sufficient
to address item (1) for our purpose.
 But we will give a 
minor  improvement using the Brun-Titchmarsh theorem,
which will be used
to derive a bound for item (2) as well.
The proof will be  essentially the same 
as in~\cite{Hastad-Schrift-Shamir}; the minor  improvement comes from 
using  the Brun-Titchmarsh theorem
and an estimate due to Rosser and Schoenfeld~\cite{Rosser-S} [Theorem 15]:
\[\frac{d}{\phi(d)} \le e^\gamma \cdot \log \log d  \cdot
\left( 1 + \frac{2.5}{e^\gamma (\log \log d)^2} \right),\]
where 
 $\phi(\cdot)$ is the Euler totient function, 
 $\gamma = 0.577\ldots$ is
Euler's constant, and $\log$ denotes
natural logarithm (as it will be for the rest of this section).
The estimate is valid for 
every $d \ge 3$, except one case $d = 2 \cdot 3 \cdot \ldots \cdot 23$
when  the constant $2.5$ should be replaced  by $2.50637$. 
We will just
use $\frac{d}{\phi(d)} \le C \log \log d$ for some universal
constant $C$, and all $d \ge 3$.

Denote by $X = 2^m -1$, $Y= \lceil \frac{X}{2} \rceil = 2^{m-1}$.

\begin{theorem}\label{HSS-period}
There exists a constant $C$, such that for any $m$ and
any randomly chosen distinct primes  $p$ and $q$ 
 of binary length $m$, $N = p\cdot q$,
and let $g$ be a randomly chosen element in  $\mathbb{Z}_N^*$,
then for all $m^2 \le A < X$,
\[{\rm Pr.}
  \left( \omega_N(g) < \frac{1}{A} \phi(N) \right)
\le C \frac{m^{2/5}}{{A}^{1/5}}, \]
where the probability 
 is over all random $Y \le p \ne q \le X$ and $g \in \mathbb{Z}_N^*$.
\end{theorem}

Note that $\phi(N) = (p-1)(q-1) \approx 2^{2m}$.
If we take $A = 2^{2\epsilon m}$ then a random $\omega_N(g)
\ge 2^{2(1-\epsilon)m}$ with probability $1 - O(m 2^{-\epsilon m/5})$.
This is more than sufficient for our required item (1) above.

The Brun-Titchmarsh theorem
is a reasonably  sharp estimate for the number of primes 
up to any upper bound $x$,
in an arithmetic progression. The bound is applicable
even when the
modulus of the arithmetic progression
 is large. The following version is an improvement of
the original Brun-Titchmarsh theorem
proved by Montgomery and Vaughan~\cite{Montgomery-Vaughan,Hooley}.
Suppose $a$ and $d$ are relatively prime.
Let $\pi(x; d, a)$ denote the number of primes $p \equiv a \bmod d$,
with $p \le x$.

\begin{theorem}[Montgomery-Vaughan]\label{B-T-thm}
\[\pi(x; d, a) \le \frac{2x}{\phi(d) \log(x/d)},\]
for all $d< x$,
where $\phi(\cdot)$ is the Euler totient function, and $\log$ denotes
natural logarithm.
\end{theorem}

Following~\cite{Hastad-Schrift-Shamir} the proof of Theorem~\ref{HSS-period}
is based on two lemmas. Let $O_N = \max\{\omega_N(x): x \in \mathbb{Z}_N^*\}$
be  the exponent of the finite Abelian group
$\mathbb{Z}_N^* \cong \mathbb{Z}_p^* \times \mathbb{Z}_q^*$, 
then $O_N = {\rm lcm}(p-1, q-1)$, and $\omega_N(x) | O_N$
for all $x \in \mathbb{Z}_N^*$.

\begin{lemma}\label{lemma1}
There exists a constant  $C_1>0$,  such that
for randomly chosen distinct primes $p$ and $q$ 
of binary length $m$, $N = p\cdot q$, and for 
any $1 \le A_1 \le X^{1/4} < 2^{m/4}$,
\[{\rm Pr.} \left( O_N < \frac{1}{A_1} \phi(N) \right)
\le C_1 \frac{1}{A_1}.\]
\end{lemma}
\proof
It is trivial if $m \le 2$. We will assume $m >2$.
Clearly $O_N = \phi(N)/{\rm gcd}(p-1, q-1)$.
So, 
\[ O_N < \frac{1}{A_1} \phi(N) \Longleftrightarrow
 {\rm gcd}(p-1, q-1) > A_1.\]
By the Prime Number Theorem, the number of primes of length $m$
is $\pi(X) - \pi(Y) \approx \frac{x}{2 \log x}$.
And so the number of ordered pairs of distinct primes of length $m$
is approximately $(\frac{x}{2 \log x})^2$.
Now we bound  the cardinality of
\[S = \{(p, q): Y \le p \ne q \le X, p, q \mbox{ are primes, and }
{\rm gcd}(p-1, q-1) > A_1\}.\]
For $p \ne q$ in that range, we claim that $p-1  \nmid q-1$.
For otherwise $(q-1)/2 \ge p-1$, which implies that
$q \ge 1 + 2(Y-1) = X$ and hence $q=X$. Then $ p-1 \le (q-1)/2 = Y -1 \le p-1$,
and so equality holds, and
  $p=Y = 2^{m-1}$, a contradiction.
It follows that ${\rm gcd}(p-1, q-1) \le (p-1)/2 < X/2$.
So, ${\rm gcd}(p-1, q-1) \le 2^{m-1} -1$.

We have
\begin{eqnarray*}
|S| & = &  \sum_{d=\lfloor A_1 \rfloor +1}^{2^{m-1} -1}
\sum_{(p, q)} {\bf 1}_{[{\rm gcd}(p-1, q-1) =d]}\\
 & \le &  \sum_{d=\lfloor A_1 \rfloor +1}^{2^{m-1} -1}
(\pi(X; d, 1) - \pi(X/2; d, 1))^2,
\end{eqnarray*}
where $\sum_{(p, q)}$ denotes the sum over
primes $(p,q)$ in the range $Y \le p \ne q \le X$.

Now we separate the sum into two parts, depending on whether
$d > \lfloor A_1^2 X^{1/3} \rfloor$.
 One part is
\[H = \sum_{d=\lfloor A_1^2 X^{1/3} \rfloor +1}^{2^{m-1} -1}
(\pi(X; d, 1) - \pi(X/2; d, 1))^2,\] 
where we use the trivial bound
$\pi(X; d, 1) - \pi(X/2; d, 1) \le \frac{X}{2d} +1$.
In the range $d< 2^{m-1}$, it is $\le \frac{X}{d}$.
It follows that
\[ H < X^2 \sum_{d=\lfloor A_1^2 X^{1/3} \rfloor +1}^{\infty}
\frac{1}{d^2} 
< \frac{X^2}{A_1^2 X^{1/3}} = \frac{X^{5/3}}{A_1^2},\]
by a comparison to the integral $\int_{K}^{\infty} \frac{1}{x^2} dx
= \frac{1}{K}$.

The other part is
\[L =  \sum_{d=\lfloor A_1 \rfloor +1}^{\lfloor A_1^2 X^{1/3} \rfloor}
(\pi(X; d, 1) - \pi(X/2; d, 1))^2,\]
where we use Theorem~\ref{B-T-thm}, to get
\[L \le \sum_{d=\lfloor A_1 \rfloor +1}^{\lfloor A_1^2 X^{1/3} \rfloor}
\left( \frac{2X}{\phi(d) \log \frac{X}{d}} \right)^2.\]
As $d \le A_1^2 X^{1/3} \le X^{5/6}$, we have $\frac{X}{d} \ge X^{1/6}$,
and $\log \frac{X}{d} \ge (\log X)/6$.
So 
\[L \le 144 \left(\frac{X}{\log X} \right)^2
\sum_{d=\lfloor A_1 \rfloor +1}^{\lfloor A_1^2 X^{1/3} \rfloor}
 \frac{1}{\phi(d)^2}.\]

Next we claim that

\noindent
{\bf Claim:} $\sum_{d>D}  \frac{1}{\phi(d)^2} = O(\frac{1}{D})$,
for any $D \ge 1$.

To prove this Claim we need a result from~\cite{Montgomery-Vaughan-book} 
[p.~61, equation (2.32)]
\[\sum_{n \le x} \left(\frac{n}{\phi(n)} \right)^2 = O(x),\]
for all $x >0$.
Let $a_n = \frac{1}{n^2}$, $b_n = \left(\frac{n}{\phi(n)} \right)^2$,
and $B_n = \sum_{k= D+1}^n b_k$,
with $n \ge D$. Then $B_{D} =0$
and $b_n = B_n - B_{n-1}$, for all $n > D$.
We have for all $Z > D$,
\[\sum_{n = D+1}^Z \frac{1}{\phi(n)^2}
= \sum_{n= D+1}^Z a_n b_n = a_Z B_Z + \sum_{n = D+1}^{Z-1}
(a_n- a_{n+1}) B_n.\]
Now $a_Z B_Z = O(1/Z)$, $a_n- a_{n+1} < 2/n^3$ and thus
$(a_n- a_{n+1}) B_n = O(1/n^2)$. It follows that
\[\sum_{n = D+1}^Z \frac{1}{\phi(n)^2}
= O(1/Z) + O(1/D).\]
Letting $Z \rightarrow \infty$ proves the Claim.

It follows that 
\[L = O \left( \left( \frac{X}{\log X} \right)^2 \cdot \frac{1}{A_1} \right).\]
And 
\[|S| \le  L + H = O \left( \left( \frac{X}{\log X} \right)^2 \cdot \frac{1}{A_1} \right) +  \frac{X^{5/3}}{A_1^2},\]
Hence,
\[{\rm Pr.} \left( O_N < \frac{1}{A_1} \phi(N) \right)
=  O \left(  \frac{1}{A_1} \right).\]
The lemma is proved.
\qed

\begin{lemma}\label{lemma2}
There exists a constant  $C_2>0$,  such that for
any $B  >1$,
\[{\rm Pr.} \left( \omega_p(g) < \frac{1}{B} \phi(p) \right)
\le C_2 \left( \frac{m}{B \log B} \right)^{1/2},\]
where the  probability
 is over a random prime $Y \le p \le X$ and
a random $g \in \mathbb{Z}_p^*$, and $\omega_p(g)$
is the order of $g$ as a group element in $ \mathbb{Z}_p^*$.
\end{lemma}

\proof
For any prime $p$, the order of  any $g \in \mathbb{Z}_p^*$
divides the order of the group $\phi(p) = p-1$,
\[|\{g \in \mathbb{Z}_p^* : \omega_p(g) < \frac{1}{B} \phi(p)\}|
= \sum_{d | p-1, ~ d  < \phi(p)/B} \phi(d).\]
Define $F(p) = \sum\limits_{d | p-1, ~ d  < \phi(p)/B} \phi(d)$ for
any prime $p$, we have
\[\sum_{Y \le p \le X} F(p) 
= \sum_{d < X/B} \phi(d) \sum_{Y \le p \le X} 
{\bf 1}_{[d | p-1]}
\le \sum_{d < X/B} \phi(d) \pi(X; d, 1).\]
Now we apply Theorem~\ref{B-T-thm} and obtain
\[\sum_{Y \le p \le X} F(p) 
\le   \sum_{d < X/B} \frac{2X}{\log(X/d)} \le \frac{2X^2}{B \log B}.\]

It follows that for any $B'>0$,
\[|\{p: Y \le p \le X, ~\mbox{$p$ is a prime, and } F(p) \ge  X/B'\}|
\le  \frac{2X^2}{B \log B} \cdot \frac{B'}{X} = \frac{2XB'}{B \log B}.\]
Then, by the Prime Number Theorem,
\[{\rm Pr.} \left( F(p) \ge  X/B' \right)
\le O \left( \frac{B' \log X}{B \log B} \right).\]
Conditional on any $p$ such that $ Y \le p \le X$ and $F(p) <  X/B'$, 
the
 probability over $g \in \mathbb{Z}_p^*$
of the event $\omega_p(g) < \frac{1}{B} \phi(p)$,
 is
$\frac{F(p)}{p-1} < \frac{3}{B'}$. Thus, the conditional
probability over both $p$ and $g \in \mathbb{Z}_p^*$ given $F(p) <  X/B'$ is
\[{\rm Pr.} \left[ \omega_p(g) < \frac{1}{B} \phi(p) \middle| F(p) <  
\frac{X}{B'}
\right]
 = O \left(\frac{1}{B'} \right).\]

It follows easily that
\[{\rm Pr.} \left( \omega_p(g) < \frac{1}{B} \phi(p) \right)
=  O \left( \frac{B' \log X}{B \log B} \right)
+ O \left(\frac{1}{B'} \right).\]
Setting $B' = (B \log B/ \log X)^{1/2}$, gives the bound of
the lemma.
\qed

Now the proof of Theorem~\ref{HSS-period}
can be completed.

\noindent
{\it Proof.} [of Theorem~\ref{HSS-period}]
We will pick $A_1$ and $A_2$ such that $A = A_1A_2$, then
\[{\rm Pr.}
  \left( \omega_N(g) < \frac{1}{A} \phi(N) \right)
= {\rm Pr.}
  \left( O_N <  \frac{1}{A_1} \phi(N) \right)
+
{\rm Pr.}
  \left( \omega_N(g) < \frac{1}{A_2} O_N \right),\]
where the first expression is over primes $Y \le p\ne q \le X$
and the second expression is over $p,q$ and $g \in \mathbb{Z}_{N}^*$.
This is seen by the contrapositive:
if $ \phi(N) \le A_1 O_N$ and $O_N \le A_2 \omega_N(g)$ then
$\phi(N) \le A \omega_N(g)$. 

By Lemma~\ref{lemma1}, the first term is $O\left(\frac{1}{A_1} \right)$.

For the second term,
we know that $\omega_N(g) = {\rm lcm}(\omega_p(g), \omega_q(g))$,
as $\mathbb{Z}_N^* \cong \mathbb{Z}_p^* \times \mathbb{Z}_q^*$.
$\omega_p(g)$ is a divisor of $p-1$, and similarly for $\omega_q(g)$.
We write $\omega_p(g) = (p-1)/a$, and $ \omega_q(g) = (q-1)/b$,
then 
\[\omega_N(g) \ge \frac{ {\rm lcm}(p-1, q-1)}{ab}.\]
To see this, we take any prime $r \mid {\rm lcm}(p-1, q-1)$,
\begin{eqnarray*}
{\rm ord}_r(\omega_N(g)) 
&=& \max\{ {\rm ord}_r(p-1) - {\rm ord}_r(a),
{\rm ord}_r(q-1) - {\rm ord}_r(b) \}\\
&\ge&  \max\{ {\rm ord}_r(p-1), {\rm ord}_r(q-1) \} - 
\max \{ {\rm ord}_r(a),  {\rm ord}_r(b) \}\\
&\ge&  \max\{ {\rm ord}_r(p-1), {\rm ord}_r(q-1) \} -
{\rm ord}_r(ab)\\
&=& {\rm ord}_r({\rm lcm}(p-1, q-1)) - {\rm ord}_r(ab).
\end{eqnarray*}
It follows that, after taking $B = \sqrt{A_2}$,
\[
{\rm Pr.}
  \left( \omega_N(g) < \frac{1}{A_2} O_N \right)
\le {\rm Pr.}
  \left( \omega_p(g) <  \frac{p-1}{B}  \right)
+ {\rm Pr.}
  \left( \omega_q(g) <  \frac{q-1}{B}  \right)
= O \left( \frac{m}{B \log B} \right)^{1/2},\]
by Lemma~\ref{lemma2}.
Equalize the two error bounds we set
\[\frac{1}{A_1} \approx \left( \frac{m}{B \log B} \right)^{1/2},\]
subject to $1 \le A_1 \le X^{1/4}$, $A_1 B^2 = A$, $B > 1$,
 where $A$ is given
as $m ^2 \le A < X$. 

We can set $B = \frac{(A^{2} m)^{1/5}}{\log A}$ to achieve the bound in Theorem~\ref{HSS-period}.

\qed



We remark that, for polynomial bounded $A= m^k$,
we can choose $B$ slightly better, $B = 
\left(\frac{m^{2k+1}}{\log m} \right)^{1/5}$,
 and achieve 
\begin{theorem}\label{HSS-period-poly}
With the same setting as in Theorem~\ref{HSS-period},
for any $k  \ge 2$
\[{\rm Pr.}
  \left( \omega_N(g) < \frac{1}{m^k} \phi(N) \right)
\le O\left( \frac{1}{m^{(k-2)/5} (\log m)^{2/5}} \right), \]
where the probability
 is over all random $Y \le p \ne q \le X$ and $g \in \mathbb{Z}_N^*$.
(The constant in $O$ depends on $k$.)
\end{theorem}

Finally, to finish the proof of Theorem~\ref{main2},
we address the required  item (2), again using the Brun-Titchmarsh theorem.

For any prime $p$, we have the 
prime factorization $p-1 = 2^{e_0} p_1^{e_1} \cdots p_k^{e_k}$.
We have
\[{\rm Pr.}
  \left( \exists g \in \mathbb{Z}_p^* : {\rm ord}_2 (\omega_p(g)) \ge e
 \right) 
\le \frac{1}{\pi(X) - \pi(Y)} 
\frac{2X}{\phi(2^e) \log(X/2^e)},\]
where the
 probability
 is over a random $Y \le p \le X$.

We have $\phi(2^e) = 2^{e-1}$ for $e\ge 1$, and $\pi(X) - \pi(Y)
=\Theta(X/\log X)$.
Using the Rosser-Schoenfeld estimate again,
we have
\[{\rm Pr.}
  \left( \exists g \in \mathbb{Z}_p^* : {\rm ord}_2 (\omega_p(g)) \ge e
 \right)
\le O \left( \frac{\log X}{\log (X/2^e)} \frac{\log \log 2^e}{2^e}  \right).\]

If we set $m^c = 2^e$, then 
we get an upper bound of $O \left(  \frac{\log \log m}{m^c} \right)$,
where the constant in $O$ depends on $c$. 
Thus, for any $c>0$,
\[{\rm Pr.}
  \left( \exists g \in \mathbb{Z}_p^* : {\rm ord}_2 (\omega_p(g)) \ge 
c \log_2 m \right)
\le O\left(  \frac{\log \log m}{m^c} \right).\]
As  $\omega_N(g) = {\rm lcm}(\omega_p(g), \omega_q(g))$, it
follows that,
\[{\rm Pr.}
  \left( \exists g \in \mathbb{Z}_N^* : {\rm ord}_2 (\omega_N(g)) \ge 
c \log_2 m \right)
\le O\left(  \frac{\log \log m}{m^c} \right).\]

Since both  required items (1) and (2) are separately true with probability
approaching 1, they are jointly true with probability
approaching 1.

\end{document}